\documentclass{article}

\usepackage[english]{babel}
\usepackage{latexsym}
\usepackage{amssymb}
\usepackage{amsmath}
\usepackage{amsthm}
\usepackage{cite}    
\usepackage{ifthen}
\usepackage{xspace}
\usepackage{paralist}
\usepackage{tikz}
\usepackage[colorlinks,final,citecolor=blue]{hyperref}

\usetikzlibrary{%
  arrows,%
  positioning,%
  decorations.pathmorphing,%
  decorations.pathreplacing,%
  automata,%
}

\theoremstyle{plain}
\newtheorem{theorem}{Theorem}[]

\newtheorem{lemma}[theorem]{Lemma}
\newtheorem{corollary}[theorem]{Corollary}

\theoremstyle{definition}
\newtheorem{definition}[theorem]{Definition}

\theoremstyle{remark}
\newtheorem{remark}[theorem]{Remark}

\newcommand{\cP}{\ensuremath{\mathcal{P}}\xspace}

\renewcommand{\phi}{\varphi}

\newcommand{\sm}{\setminus}

\newcommand{\sett}[2]{\left\{#1\mathrel{\left|
	\vphantom{#1}\vphantom{#2}\right.}#2\right\}}
\newcommand{\bsett}[2]{\Bigl\{#1\mathrel{\Bigl|\Bigr.}#2\Bigr\}}
\newcommand{\set}[1]{\left\{\mathinner{#1}\right\}}

\newcommand{\abs}[1]{\left|\mathinner{#1}\right|}

\newcommand{\N}{\mathbb{N}}

\newcommand{\xra}[1]{\xrightarrow{#1}}
\newcommand{\xras}[1]{\xrightarrow{#1}\!\!{}^*\,}

\newcommand{\e}{\lambda}
\newcommand{\ew}{\e}

\newcommand\tsid{\trt_{k}}
\newcommand\tsidp[1]{\trt_{#1}}

\newcommand\dsid{\ensuremath{\delta}\xspace}

\newcommand{\trt}{\mathbf{t}}   
\newcommand{\etr}{\trt^{\mathrm{e}}}
\newcommand{\epath}{P^{\mathrm{e}}}
\newcommand{\etsid}{\tsid^{\mathrm{e}}}
\newcommand{\cpt}{\cP_{\trt}}   

\newcommand\al{\Sigma}        
\newcommand\alG{\Gamma}        
\newcommand\eew{(\e/\e)}        
\newcommand\ealph{E_\al}   
\newcommand\aut{\mathbf{a}}   
\newcommand\autb{\mathbf{b}}   
\newcommand\tr{\mathbf{t}}    
\newcommand\sz[1]{|#1|}       
\newcommand\weight[1]{\mathrm{weight}(#1)}       
\newcommand\ch{\gamma}        
\newcommand\chsid{\mathrm{sid}} 
\newcommand\dist{\dsid}               
\newcommand\inp{\mathrm{inp}}
\newcommand\out{\mathrm{out}}
\newcommand\db{B}             
\newcommand\dbold{D}   
\newcommand\pssi{\par\smallskip\indent}
\newcommand\pmsi{\par\medskip\indent}

\newcommand\pssn{\par\smallskip\noindent}
\newcommand\pmsn{\par\medskip\noindent}
\newcommand\pbsn{\par\bigskip\noindent}
\newcommand\pnsi{\par\indent}

\begin{document}

\begin{center}
\textbf{\Large An efficient algorithm for computing the edit distance of a regular language via input-altering transducers}
\end{center}

{\large Lila Kari$^{1}$, Stavros Konstantinidis$^{2}$,
Steffen Kopecki$^{1,2}$, Meng Yang$^{2}$}
\pmsn
$^{1}$ Department of Computer Science, University of Western Ontario, London, Ontario, Canada,
\texttt{lila@csd.uwo.ca, steffen@csd.uwo.ca}
\pssn
$^{2}$ Department of Mathematics and Computing Science,
Saint Mary's University, Halifax, Nova Scotia, Canada,
\texttt{s.konstantinidis@smu.ca, meyang.mike@gmail.com}

\pbsn
\textbf{Abstract.}
We revisit the problem of computing the edit distance
of a regular language given via an NFA. This problem
relates to the inherent maximal error-detecting capability of the language in question. We present an efficient algorithm
for solving this problem which executes in time $O(r^2n^2d)$, where $r$ is the cardinality of
the alphabet involved, $n$ is
the number of transitions in the given NFA, and $d$ is the
computed edit distance. We have implemented the algorithm
and present here performance tests. The correctness of the
algorithm is based on the result (also presented here) that
the particular error-detection property related to our problem can be defined via an input-altering transducer.

\pbsn
\textbf{Keywords}.
algorithms, automata, complexity, edit distance, implementation, transducers, regular language

\section{Introduction}\label{sec:intro}
The edit distance of a language $L$ with at least two
words---also referred to as
inner edit distance of $L$---is the minimum edit distance
between any two different words in $L$.
In \cite{Kon:2007}, the author considers the problem of
computing the edit distance of a regular language, which
is given via a nondeterministic finite automaton (NFA),
or a deterministic finite automaton (DFA).
For a given automaton $\aut$ with
$n$ transitions and an alphabet of $r$ symbols, the algorithm proposed in \cite{Kon:2007} has
worst-case time complexity
\begin{equation}\label{eq:cxty}
O(r^2n^2q^2(q+r)),
\end{equation}
where, in fact, $q$ is either the number of states in $\aut$ (if $\aut$ is a DFA), or the square of the number of states in
$\aut$ (if $\aut$ is an NFA).
If the size of the alphabet is ignored and the automaton in question has only states that can be reached from the start state, then the number of states
is $O(n)$ and the worst-case time complexity shown in~(\ref{eq:cxty})
can be written as
\begin{equation}\label{eq:cxty2}
O(n^5) \hbox{ for DFAs},\> \hbox {and }\> O(n^8) \hbox{ for NFAs}.
\end{equation}
\par
In this paper, motivated by the question of whether certain
error-detection properties can be defined via input-altering
transducers, we obtain an efficient algorithm to compute the
edit distance of a regular language given via an NFA with
$n$ transitions---see theorem~\ref{th:final}. The
algorithm, which is called \texttt{DistBestInpAlter}, has worst-case time complexity
\begin{equation}\label{eq:cxty3}
O(n^2d),
\end{equation}
where $d$ is the computed distance, which is a significant improvement over the original algorithm in \cite{Kon:2007}.

We note that an approach of computing the edit distance
problem via the error-detection property is
discussed briefly in\cite{KonSil:2010}. A similar approach
can be used for the edit distance problem via the
error-correction property.
The new algorithm \texttt{DistBestInpAlter}---see theorem~\ref{th:final}---is based on
(a) the new result that the error-detection property related to our problem is definable via an efficient input altering transducer---see theorem~\ref{th:iat:ed}, and (b) the
observation that the preliminary error-detection-based algorithm can be made significantly more efficient by
a nontrivial utilization of the above new result.
For clarity
of presentation we present in detail not only the
new algorithm, but also
the intermediate versions, all of which
have been implemented  in Python using the well maintained library
FAdo for automata \cite{Fado}.
We have also tested all versions experimentally,
and we discuss in this paper the outcomes of the tests showing that, not only in theory, but also in practice algorithm \texttt{DistBestInpAlter} is clearly more performant.

We note that some related problems involving distances between
words and languages can be found in
\cite{Wag:1974,Pigh:2001} (edit distance between a word and a language), and in
\cite{Mohri:2003,KKPWX:2003,BPR:icalp2011,HKS:dlt2012,HKS:ciaa2013} (various distances between languages).
The problem considered here is technically different,
as the desired distance involves
different words within the same language.

The paper is organized as follows.
The next section contains basic notions
on languages, finite-state machines and edit-strings, and
a few preliminary lemmata.
Section~\ref{sec:ed} describes the approach of
computing the desired edit distance via the concepts
of error-detection and -correction.
Section~\ref{sec:iat} first presents the new result that
the error-detection property in question is definable via
an efficient input-altering transducer---see
theorem~\ref{th:iat:ed}---and then, the main result, algorithm \texttt{DistBestInpAlter} in theorem~\ref{th:final}.
Section~\ref{sec:implem} discusses the implementation
and testing of the main algorithm and its intermediate versions. The last section
contains a few concluding remarks and questions
for future research.

\section{Notation, background and preliminary results}\label{sec:two}
Most of the basic notions presented here can be found
in various texts such as
\cite{Be:1979,Wood:1987,FLhandbookI,Yu:handbook,Sak:2009}.

\subsection{Sets, words, languages, channels}
If $S$ is any set, the expression $|S|$
denotes the cardinality
of $S$.
When there is no risk of confusion we denote a singleton set $\{u\}$ simply as $u$. For example, $S\cup u$ is the union of $S$ and $\{u\}$.
We use standard basic notation and terminology for alphabets, words and languages---see \cite{MaSa:handbook}, for instance. For example, $\al$ denotes an alphabet, $\al^+$ the set of nonempty words, $\ew$ the empty word, $|w|$ the length of the word $w$.
We use the concepts of (formal) language and concatenation between words, or languages, in the usual way. We say that $w$ is an $L$-\emph{word} if $w\in L$ and $L$ is a language.
\par
A binary word \emph{relation} $\rho$ on $\al^*$ is any subset of $\al^*\times\al^*$. The \emph{domain} of $\rho$ is $\{u\mid(u,v)\in\rho \hbox{ for some $v\in\al^*$}\}$.
%
A \emph{channel} $\ch$ is a binary relation on $\al^*$ that is domain-preserving (or input-preserving); that is, $\ch\subseteq\al^*\times\al^*$ and $(w,w)\in\ch$ for all words $w$
in the domain of $\ch$. When  $(u,v)\in\ch$ we say that $u$ can be received as $v$
via the channel $\ch$, or $v$ is a possible output of $\ch$
when $u$ is used as input. If $v\not=u$ then we say that
$u$ can be received with errors (via $\ch$). Here we only
consider the channel $\chsid(k)$, for some $k\in\N$, such that
$(u,v)\in\chsid(k)$ if and only if $v$ can be obtained by
applying at most $k$  errors in $u$, where an error
could be a deletion of a symbol in $u$, a substitution of a symbol in $u$ with
another symbol, or an insertion of a symbol in $u$---see further below
for a more rigorous definition via edit-strings.

\subsection{NFAs and transducers}
A  nondeterministic finite automaton
with empty transitions, \emph{$\ew$-NFA} for short, or
just \emph{automaton}, is a quintuple $\aut=(Q,\al, T,s, F)$ such that $Q$ is the set of states, $\al$ is the alphabet, $s\in Q$
is the start (or initial) state, $F\subseteq Q$ is the set of final states, and $T\subseteq Q\times(\al\cup\ew)\times Q$ is the finite set of transitions. Let $(p,x,q)$ be a transition of $\aut$. Then  $x$ is called the \emph{label} of the transition, and we say that $p$ has an \emph{outgoing} transition (with label $x$).
We also use the notation $$p\xra{x}q$$ for a transition
$(p,x,q)$.
The $\ew$-NFA $\aut$ is called an \emph{NFA}, if no transition label is empty, that is, $T\subseteq Q\times\al\times Q$. A deterministic finite automaton, \emph{DFA} for short, is a special type of NFA where there is no state $p$ having two outgoing transitions with different labels.

\par
A \emph{path} 
of $\aut$ is a finite sequence of transitions of the form
 $$(p_0,x_1,p_1),(p_1,x_2,p_2),\ldots,
 (p_{\ell-1},x_\ell,p_\ell),$$
for some nonnegative integer $\ell$. The word $x_1\cdots x_\ell$ is called the \emph{label} of the path.
We write $p_0\xras{x}p_\ell$ to
indicate that there is a path with label $x$ from $p_0$ to
$p_\ell$.
A path
as above is called \emph{accepting} if
$p_0$ is the start state and $p_\ell$ is a final state. The \emph{language accepted} by $\aut$, denoted as $L(\aut)$, is the set of labels of all the accepting paths of $\aut$. The automaton $\aut$ is called
\emph{trim}, if every state appears in some accepting path of $\aut$.
\par
A (finite) \emph{transducer} \cite{Be:1979,Yu:handbook} is a sextuple $\tr=(Q, \al, \alG, T, s, F)$ such that $Q,s,F$ are exactly the same as those in $\ew$-NFAs, $\al$ is now called the input alphabet, $\alG$ is the output alphabet, and $T\subseteq Q\times\al^*\times\alG^*\times Q$ is the finite set of transitions. We write $(p,u/v,q)$,
or $p\xra{u/v}q$ for a transition---the label here is $(u/v)$, with $u$ being the input and $v$ being the output label. The concepts of path, accepting path, and trim transducer are similar to those in $\ew$-NFAs. However, the label of a
transducer path
$(p_0,x_1/y_1,p_1),\ldots,(p_{\ell-1},x_\ell/y_\ell,p_\ell)$
is the pair $(x_1\cdots x_\ell,
y_1\cdots y_\ell)$ of the two words consisting of the input and output labels in the path, respectively. The \emph{relation realized}
by the transducer $\tr$, denoted as $R(\tr)$, is the set of labels in all the accepting paths of $\tr$.
We write $\tr(x)$ for the set of possible outputs of $\tr$ on input $x$, that is,  $y\in\tr(x)$ if and only if $(x,y)\in R(\tr)$.
The transducer is called functional, if
the relation $R(\tr)$ is a function, that is,
$\tr(x)$ consists of at most one word, for all inputs $x$.
The transducer $\tr$ is said to be in \emph{standard form}, if each transition $(p,u/v,q)$ is such that $u\in(\al\cup\ew)$ and $v\in(\alG\cup\ew)$.
We note that every transducer is effectively equivalent to one (realizing the same relation, that is) in standard form.
\par
If $\mathbf{m}$ is an automaton, or a transducer in standard form, then the \emph{size} of
$\mathbf{m}$, denoted by $\sz{\mathbf{m}}$, is the number of states plus the number
of transitions in $\mathbf{m}$.

\subsection{Edit strings and edit distance.}
The alphabet $\ealph$ of the
\emph{(basic) edit operations}, which depends on the alphabet $\al$
of ordinary symbols, consists of all
symbols $(x/y)$ such that $x,y\in\al\cup\{\e\}$ and at least one
of $x$ and $y$ is in $\al$. If $(x/y)\in\ealph$ and $x$ is
not equal to $y$ then $(x/y)$ is called an \emph{error} \cite{KaKo:2004}.
The edit operations $(a/b)$, $(\e/a)$, $(a/\e)$,
where $a,b\in\al-\{\e\}$ and $a\not=b$,
are called \emph{substitution, insertion, deletion}, respectively.
We write $\eew$ for the empty word over the alphabet $\ealph$. We
note that $\e$ is used as a formal symbol in the elements of
$\ealph$. For example, if $a,b\in\al$ then
$(\e/a)(b/b)\not=(b/a)(\e/b)$.
The elements of $\ealph^*$ are
called \emph{edit strings}.
The {\it weight\/} of an edit string $h$, denoted as $\weight{h}$,  is the number
of errors occurring in $h$. For example, for
\begin{equation}\label{eq:edits}
g = (a/a)(a/\e)(b/b)(b/a)(b/b),
\end{equation}
$\weight{g}=2$.
The \emph{input} and \emph{output} parts of an edit
string $h=(x_1/y_1)\cdots(x_n/y_n)$ are the words (over $\al$)
$x_1\cdots x_n$ and $y_1\cdots y_n$, respectively. We write
$\inp(h)$ for the input part and $\out(h)$ for the output part of
$h$. For example, for the $g$ shown above,
$\inp(g)=aabbb$ and $\out(g)=abab$.
The \emph{inverse of an edit string} $h$ is the edit string
resulting by inverting the order of the input and output
parts in every edit operation in $h$. For example, the inverse
of $g$ shown above is
\[
(a/a)(\e/a)(b/b)(a/b)(b/b).
\]
The channel $\chsid(k)$ can be defined more rigorously via edit
strings:
\[
\chsid(k)=\{(u,v)\mid \> u=\inp(h),\>v=\out(h),\>\hbox{for some $h\in\ealph^*$ with $\weight{h}\le k$}\}.
\]
\pssn
The \emph{edit (or Levenshtein) distance} \cite{Levenshtein:66:en} between two words $u$ and
$v$, denoted by $\dist(u,v)$, is the smallest number of errors
(substitutions, insertions and deletions) that can be used to
transform $u$ to $v$. More formally,
\[
\dist(u,v)=\min\{\weight{h}\mid h\in
\ealph^*,\>\inp(h)=u,\>\out(h)=v\}.
\]
We say that an edit string $h$ \emph{realizes the edit distance
between two words} $u$ and $v$, if $\weight{h}=\dist(u,v)$ and
$\inp(h)=u$ and $\out(h)=v$.
For example, for $\al=\{a,b\}$, we have that
$\dist(ababa,babbb)=3$ and the edit string
\[
h=(a/\e)(b/b)(a/a)(b/b)(a/b)(\e/b)
\]
realizes $\dist(ababa,babbb)$.
Note that several edit strings can realize the distance $\delta(u,v)$.
If $L$ is a language containing at least two words then the edit
distance of $L$ is
\[
\dist(L)=\min\{\dist(u,v)\mid u,v\in L\>\hbox{ and }\>u\not=v\}.
\]
Testing whether a given NFA accepts at least two words is
not a concern in this paper, but we note that this can
be done efficiently (in linear time via a breadth
first search type algorithm) \cite{Yang:2012}.

\begin{definition}
An edit string $h$ of nonzero weight is called \emph{reduced}, if
(a) the first error in $h$ is not an insertion, and
(b) if the first error in $h$ is a deletion of the form $(a/\e)$, then the first non-deletion edit operation $(x/y)$ that follows $(a/\e)$ in $h$ (if any) is such that $y\not=a$.
\end{definition}

\begin{lemma}\label{lem:didist}
Let $x,y,u,v$ be words. The following statements hold true.
\begin{enumerate}
\item
$\dsid(xuy,xvy) = \dsid(u,v)$.
\item
If $v<_p u$ then $\dsid(u,v)=\abs u-\abs v$.
\item
If $u\not=v$, then there is a reduced edit string $h$
realizing $\dsid(u,v)$.
\end{enumerate}
\end{lemma}
\begin{proof}
The first statement already appears in \cite{Levenshtein:66:en}. The second statement
is rather folklore, but we provide a proof here for the sake of completeness.
Let $u=\sigma_1\cdots\sigma_n$ and $v=\sigma_1\cdots\sigma_m$, where $m,n\in\N_0$ and $m<n$ and all $\sigma_i$'s are in $\al$. Then, the edit
string
\[
h=(\sigma_1/\sigma_1)\cdots(\sigma_m/\sigma_m)
(\sigma_{m+1}/\e)\cdots(\sigma_n/\e)
\]
has weight $n-m$ and $\inp(h)=u$ and $\out(h)=v$.
We show that $h$ realizes $\dsid(u,v)$ by proving that,
for any edit string $g$ realizing $\dsid(u,v)$,
$\weight{g}=n-m$. Indeed, first note that
$\weight{g}\le \weight{h}=n-m$. Let $i$ and $d$ be the number
of insertions and deletions in $g$. Then $|v|=|u|+i-d$,
which implies $n-m=d-i$. Now $\weight{g}\ge d+i\ge
d-i= n-m$, as required.
\par
For the third statement, let $g_0$ be any edit string
realizing $\dsid(u,v)$. The following process can be used
to obtain the required reduced edit string $h$.
\begin{enumerate}
  \item If the first error in $g_0$ is a substitution, then $h=g_0$.
  \item If the first error in $g_0$ is an insertion, then  set $g_0$ to the inverse of $g_0$ and continue with
      the next step.
  \item If the first error in $g_0$ is a deletion $(a/\e)$, then $g_0$ is of the form
      \[
      g_0=(e_1\cdots e_r)(a/\e)(a_1/\e)\cdots(a_d/\e)g_0',
      \]
      where the $e_i$'s are non-errors, $d\in\N_0$ and each
      $(a_j/\e)$ is a deletion, and $g_0'$
      does not start with a deletion. If $g_0'$ is empty or starts with an edit operation $(x/y)$ in which $y\not=a$, then the required $h$ is $g_0$. If $g_0'$
      starts with an edit operation $(x/a)$, then it is of the form $g_0'=(x/a)g_1'$, and the
      edit string
      \[
      g_1=(e_1\cdots e_r)(a/a)(a_1/\e)\cdots(a_d/\e)(x/\e)g_1',
      \]
      realizes $\dsid(u,v)$, as  $weight(g_1) = weight(g_0)$. The process now continues from the first step using $g_1$ for $g_0$.
\end{enumerate}
As the edit string $g_0$ is finite, the above process
terminates with a reduced edit string $h$, as required.
\end{proof}

The bound $\dbold_\aut$ in the next lemma comes from
\cite{Kon:2007}. It is always less than or equal to
the number of states in the NFA $\aut$. Moreover, there are NFAs for which this bound is tight---see Fig.~\ref{fig:testA}
in Section~\ref{sec:implem}.

\begin{lemma}
For every NFA $\aut$ accepting at least two words
we have that $$\dist(L(\aut))\le \dbold_\aut,$$
where $\dbold_\aut$ is the number of states in the
longest path in $\aut$ from the start state having no repeated state.
\end{lemma}

However, the bound $\dbold_\aut$ is of no use in our context,
as the problem of determining the length of a longest path in a given automaton, or a graph in general, is NP-complete since an algorithm solving this problem can be used to decide the existence of a Hamiltonian path; see for example \cite{Schr:2003}.
There are many ways to obtain an efficiently computable
upper bound on the edit distance of $L(\aut)$ that is always
at most equal to the number of states in $\aut$.
For example, that distance is always less than are equal
to the distance of  two shortest accepted words.
We agree to use this as a working upper bound:

\begin{lemma}\label{lem:bound}
For every NFA $\aut$ accepting at least two words
we have that $$\dist(L(\aut))\le \db_\aut,$$
where $\db_\aut$ is the edit distance of two
shortest words in $L(\aut)$.
\end{lemma}

\section{Edit distance via error-detection and -correction}\label{sec:ed}
In \cite{KonSil:2010}, the authors discuss a
conceptual method for computing integral distances of regular languages---integral means that all distance values are positive integers---via
the property of error-detection.
In this section, we review that method and produce
a \emph{concrete} preliminary algorithm for computing
the edit distance of a regular language.
We also present here a similar method, via the property of error-correction, and the algorithm it entails. In fact this
latter algorithm \emph{estimates} the edit distance, as it returns two integers, differing by 1, one of which is the exact edit distance value.
Both algorithms have been implemented as
will be discussed
in section~\ref{sec:implem}.

A language $L$ is \emph{error-detecting for} a channel $\ch$, if no $L$-word can be received as a different $L$-word via $\ch$, that is,
for any words $u$ and $v$,
\[
u,v\in L\>\hbox{and}\> (u,v)\in\ch\>\>\rightarrow\>\> u=v
\]
\textbf{Note}: The definition of error-detection in
\cite{Kon:2002} uses $L\cup\{\e\}$ instead of $L$
in the above formula. This slight change makes the
presentation here simpler and has no bearing on any
existing results regarding error-detecting languages.

A language $L$ is \emph{error-correcting for} a channel $\ch$, if no two different $L$-words can result into the same word via $\ch$, that is,
\[
u,v\in L\>\hbox{and}\> (u,z),(v,z)\in\ch\>\>\rightarrow\>\> u=v
\]
This property of $L$ ensures that any output $z$ of the
channel can be corrected to a unique $L$-word.

\begin{remark}\label{rem:ed}
The error-detection method of \cite{KonSil:2010}, as well as the error-correction method, are based on the
following observations, where $\aut$ is an NFA
and $\tr$ is an input-preserving transducer.
\begin{enumerate}
  \item
   A language $L$ is error-detecting for $\chsid(m)$,
   if and only if
   $\dist(L)>m$.
  \item
   A language $L$ is error-correcting for $\chsid(k)$,
   if and only if
   $\dist(L)>2k$,  \cite{Levenshtein:66:en}.
  \item
  A language $L$ is error-detecting for a channel
  $\ch$ if and only if the relation
  \begin{equation}\label{eq:rel-ed}
    \ch\cap(L\times \al^*)\cap(\al^*\times L)
  \end{equation}
  is functional \cite{Kon:2002}.
  \item
  A language $L$ is error-correcting for a channel
  $\ch$ if and only if the relation
  \begin{equation}\label{eq:rel-ec}
      \ch^{-1}\cap(\al^*\times L)
  \end{equation}
  is functional \cite{Kon:2002}.
  \item
  Suppose $\aut$ accepts $L$ and $\tr$ realizes $\ch$.
  A transducer, denoted as $(\tr\downarrow\aut\uparrow\aut)$,
  that realizes relation~(\ref{eq:rel-ed}) can be constructed
  in time $O(\sz{\tr}\sz{\aut}^2)$.
  Moreover, a transducer, denoted as $(\tr^{-1}\uparrow\aut)$,
  that realizes relation~(\ref{eq:rel-ec}) can be constructed
  in time $O(\sz{\tr}\sz{\aut})$ \cite{Kon:2002}.
  \item
  There is a quadratic time algorithm that decides whether a given transducer is functional \cite{AllMoh:2003,BeCaPrSa2003}.
\end{enumerate}
\end{remark}
Using the above observations, we present first
the error-detection-based algorithm for computing
the desired edit distance, and further below the algorithm based on error-correction.
\begin{figure}[ht]
\begin{tabbing}
PAR \= No \= No \= NN \= NN \=\kill
\> Algorithm \texttt{DistErrDetect} \\
\> 0.\> Input: NFA $\aut$ \hspace{4mm} \\
\> 1.\> Let ${\db_\aut}$ be edit distance bound in Lemma~\ref{lem:bound}\\
\> 2.\> Let $\min\leftarrow 1$ and
$\max\leftarrow{\db_\aut}-1$ \\
\> 3.\> Perform binary search to find the largest $k$ in
      $\{\min,\ldots,\max\}$ \\
\>  \> for which $L(\aut)$ is error-detecting for $\chsid(k)$ as follows: \\
\>  \>  \textbf{while} ($\min\le \max$)\\
\>  \>  a)\> Let $k\leftarrow\lfloor(\min+\max)/2\rfloor$\\
\>   \> b)\> Construct transducer $\tr$ realizing the channel $\chsid(k)$---see Fig.~\ref{fig:inprestrans}\\
\>   \> c)\> Construct the transducer $\tr'\leftarrow(\tr\downarrow\aut\uparrow\aut)$\\
\>   \> d)\> If ($\tr'$ is functional)  let $\min\leftarrow k+1$\\
\>   \> \>  Else  let $\max\leftarrow k-1$\\
\> 4. \> \textbf{return} $\min$
\end{tabbing}
\end{figure}

\begin{figure}[ht]
\centering
\begin{tikzpicture}[shorten >=1pt,node distance=2.5cm,
		on grid,>=stealth',initial text=,
		every state/.style={inner sep=2pt, minimum size=.8cm},
		 phantom/.style={draw=white},font=\footnotesize]
	\node[state,initial,accepting] (q0) {$[0]$};
	\node[state,accepting] (p1) [right=of q0] {$[1]$};
	\node[state,accepting] (p2) [right=of p1] {$[2]$};
	\node[state,phantom] (p3) [right=of p2] {$\cdots$};
	\node[state,accepting] (p4) [right=of p3] {$[k]$};

	\path[->] (q0) edge node [below] {$\sigma\slash\lambda\quad\lambda\slash\sigma$}
                node [above] {$\sigma\slash\tau$} (p1)
            edge [loop below] node {$\sigma\slash\sigma$} ()
		 (p1) edge node [below] {$\sigma\slash\lambda\quad\lambda\slash\sigma$}
				node [above] {$\sigma\slash\tau$} (p2)
			edge [loop below] node {$\sigma\slash\sigma$} ()
		(p2) edge node [below] {$\sigma\slash\lambda\quad\lambda\slash\sigma$}
				node [above] {$\sigma\slash\tau$} (p3)
			edge [loop below] node {$\sigma\slash\sigma$} ()
		(p3) edge node [below] {$\sigma\slash\lambda\quad\lambda\slash\sigma$}
				node [above] {$\sigma\slash\tau$} (p4)
		(p4) edge [loop below] node {$\sigma\slash\sigma$} ();
\end{tikzpicture}
\parbox{4.3in}{\caption{An input-preserving transducer realizing the channel $\chsid(k)$. Each edge label $\sigma/\sigma$ represents many transitions, one for each symbol $\sigma$ of the alphabet, and similarly for $\sigma/\lambda$ and $\lambda/\sigma$. Each edge label $\sigma/\tau$ represents
many transitions, one for each pair of distinct symbols
$\sigma$ and $\tau$ from the alphabet. Thus, if the alphabet size is $r$, then the size of the transducer is $O(r^2k)$, as $r,k\to\infty$, or simply $O(k)$ if $r$ is fixed.}\label{fig:inprestrans}}
\end{figure}
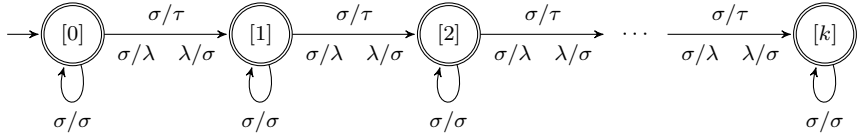

\begin{corollary}\label{th:ed}
Algorithm \texttt{DistErrDetect} computes the
edit distance of a language given via an NFA $\aut$ in time $$O(\sz{\aut}^4r^4\db_{\aut}^2\log\db_{\aut}),$$
where $r$ is the cardinality of the alphabet used in $\aut$.
\end{corollary}
\begin{proof}
For the correctness of the algorithm, first note that the loop in step 3 is set up such that $L(\aut)$
is always error-detecting for $\chsid(\min-1)$. Also, based on the observations listed in the above remark, if $L(\aut)$ is error-detecting for $\chsid(k)$ but not for
$\chsid(k+1)$, then the desired distance
must be greater than $k$ and at most $k+1$, hence equal
to $k+1$.
\pnsi
For the time complexity, the while loop will
perform $O(\log {\db_\aut})$ iterations. In each
iteration, the value $k$ is
used to construct the transducer of
size $O(r^2k)$ shown in Fig.~\ref{fig:inprestrans} with alphabet
being the set of alphabet symbols appearing
in the description of $\aut$. Then, the transducer
$\tr'$ is constructed and its functionality
is tested in time $O(\sz{\aut}^4r^4k^2)$.
As $k<{\db_\aut}$, it follows that the total time complexity is
as required.
\end{proof}

We note that, in the worst case, $\db_{\aut}$ is
of order $O(\sz{\aut})$ and, assuming a fixed alphabet, the above algorithm operates in time
$$O(\sz{\aut}^6\log\sz{\aut}),$$
which  is asymptotically better
than the time complexity stated in \cite{Kon:2007} when the
given automaton is an NFA.

\pssi
Next we present
the error-correction-based algorithm for estimating
the desired edit distance.

\begin{figure}[ht]
\begin{tabbing}
PAR \= No \= No \= NN \= \kill
\> Algorithm \texttt{DistErrCorrect} \\
\> 0.\> Input: NFA $\aut$ \hspace{4mm} \\
\> 1.\> Let ${\db_\aut}$ be the bound in Lemma~\ref{lem:bound}\\
\> 2.\> Let $\min\leftarrow 1$ and $\max\leftarrow\lfloor({\db_\aut}-1)/2\rfloor$\\
\> 3.\> Perform binary search to find  the largest $k$ in
      $\{\min,\ldots,\max\}$ \\
\>  \> for which $L(\aut)$ is error-correcting for $\chsid(k)$ as follows: \\
\>  \>  \textbf{while} ($\min\le \max$)\\
\>  \>  a)\> Let $k\leftarrow\lfloor(\min+\max)/2\rfloor$\\
\>   \> b)\> Construct a transducer $\tr$ realizing the channel $\chsid(k)$\\
\>   \> c)\> Construct the transducer $\tr'\leftarrow(\tr^{-1}\uparrow\aut)$\\
\>   \> d)\> If ($\tr'$ is functional)  let $\min\leftarrow k+1$\\
\>   \> \>  Else  let $\max\leftarrow k-1$\\
\> 4. \>  \textbf{return} $\{2\min-1,2\min\}$
\end{tabbing}
\end{figure}

\begin{corollary}\label{cor:ec}
Algorithm \texttt{DistErrCorrect} returns
two values, differing by 1, one of which is the edit distance of
the language given via $\aut$, in time   $$O(\sz{\aut}^2r^4(\db_{\aut}/2)^2\log(\db_{\aut}/2)),$$
where $r$ is the cardinality of the alphabet used in $\aut$.
\end{corollary}
\begin{proof}
For the correctness of the algorithm, first note that the loop in step 3 is set up such that $L(\aut)$
is always error-correcting for $\chsid(\min-1)$. Also, based on the observations listed in the above remark, if $L(\aut)$ is error-correcting for $\chsid(k)$ but not for
$\chsid(k+1)$, then the desired distance
must be greater than $2k$ and at most $2(k+1)$, hence equal
to $2k+1$, or $2k+2$. Moreover, as ${\db_\aut}\ge2k+1$,
the initial value of $\max$ in step 2 is correct.
\pnsi
For the time complexity, the while loop will
perform $O(\log {\db_\aut})$ iterations. In each
iteration, the value $k$ is
used to construct the transducer of
size $O(r^2k)$ shown in Fig.~\ref{fig:inprestrans} with alphabet
being the set of alphabet symbols appearing
in the description of $\aut$. Then, the transducer
$\tr'$ is constructed and its functionality
is tested in time $O(\sz{\aut}^2r^4k^2)$.
As $k<{\db_\aut}/2$, it follows that the total time complexity is
as required.
\end{proof}
As noted before, in the worst case, $\db_{\aut}$ is
of order $O(\sz{\aut})$ and, assuming a fixed alphabet, the above algorithm operates in time
$$O(\sz{\aut}^4\log\sz{\aut}).$$
This time complexity is asymptotically better than the one in \cite{Kon:2007} even when the given automaton is a DFA.

\section{An $O(n^2d)$ algorithm for edit distance via input-altering transducers}\label{sec:iat}
In this section we present a new (exact) method for computing much faster the desired edit distance via input-altering transducers---see theorem~\ref{th:final} and the associated algorithm.
A transducer $\tr$ is called \emph{input-altering}, if
\[
w\notin\tr(w),\>\hbox{ for all words $w$,}
\]
that is, the output of $\tr$ is never equal to the input used. The new method is based on the following two major observations.
\begin{description}
  \item[(4.1)] The new result (see theorem~\ref{th:iat:ed}) that the property of error-detection for the channel $\chsid(k)$ can be described via an input-altering transducer $\tr_k$ of size $O(k)$.
  \item[(4.2)] The new observation that, using an input-altering transducer in our algorithms, eliminates the need for a binary search loop that builds a new transducer in each iteration. Instead, this loop can be replaced with the incremental construction of an NFA $\aut'_k$, which depends on $\tr_k$, until a certain condition is satisfied, in which case the value of $k$ is the desired edit distance---see further below for details.
\end{description}
The above observations are presented in two subsections.

\subsection{An input-altering transducer for error-detection}
We give first a quick summary of some concepts discussed
in \cite{DudKon:2012}.

\begin{remark}\label{rem:iat}
Let $\tr$ be an input-altering transducer. The \emph{property
$\cpt$ described by} $\tr$ is the set of all languages $L$ satisfying
\begin{equation}\label{eq:inpalt}
\tr(L)\cap L=\emptyset.
\end{equation}
As explained in \cite{DudKon:2012}, this concept constitutes a formal method for specifying certain code properties defined via abstract
binary relations \cite{Shyr:Thierrin:relations}, and allows one to decide efficiently the
\emph{property satisfaction problem} by testing
condition~(\ref{eq:inpalt}). In particular, condition~(\ref{eq:inpalt}) can be tested
in time
\begin{equation}\label{eq:cxty4}
O(\sz{\aut}^2\sz{\tr}),
\end{equation}
where $\aut$ is the NFA accepting the language $L$
and $\tr$ is the input-altering transducer
describing the property for which $L$ is to be tested.
This approach has led to the development of an
online language  server, called I-LaSer \cite{ilaser}.
\end{remark}

We shall show (see theorem~\ref{th:iat:ed}) that error-detection for $\chsid(k)$ is definable via the input-altering transducer
$\tsid$, which is shown in Fig.~\ref{fig:inaltertrans} and
defined next.
The value $i$ in a state $[i]$ or $[i,a]$ is called the {\em error counter}, meaning that any path from $[0]$ to a state with error counter $i$ has to be labeled $u\slash v$ such that $\dsid(u,v) \le i$.
More precisely, we will define the edges such that a state $[i,a]$ can be reached from $[0]$ via a path with label  $u\slash v$ if and only if $u = vax$ for some word $x$ and $i = \abs{ax}$, thus, $v$ is a proper prefix of $u$ and state $[i,a]$ remembers the left-most letter of $u$ that occurs after its prefix $v$.
A state $[i]$ with $i\ge 1$ can only be reached via a path labeled $u\slash v$ from $[0]$ if
$1\le \dsid(u,v) \le i$, thus, $u\neq v$.
Furthermore, we make sure that for $u\neq v$ such that neither $u\le_p v$ nor $v\le_p u$ there is a path from $[0]$ to $[\dsid(u,v)]$ which is labeled by $u\slash v$ or $v\slash u$.

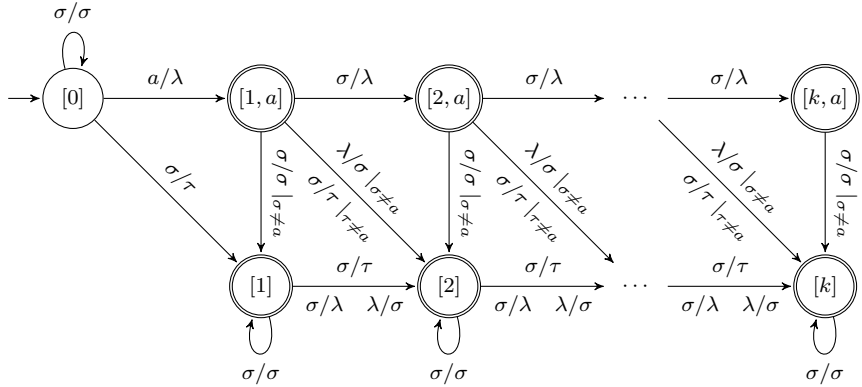
\begin{figure}[ht]
\centering
\begin{tikzpicture}[shorten >=1pt,node distance=2.5cm,
		on grid,>=stealth',initial text=,
		every state/.style={inner sep=2pt, minimum size=.8cm},
		 phantom/.style={draw=white},font=\footnotesize]
	\node[state,initial] (q0) {$[0]$};
	\node[state,accepting] (q1) [right=of q0] {$[1,a]$};
	\node[state,accepting] (q2) [right=of q1] {$[2,a]$};
	\node[state,phantom] (q3) [right=of q2] {$\cdots$};
	\node[state,accepting] (q4) [right=of q3] {$[k,a]$};

	\node[state,accepting] (p1) [below=of q1] {$[1]$};
	\node[state,accepting] (p2) [right=of p1] {$[2]$};
	\node[state,phantom] (p3) [right=of p2] {$\cdots$};
	\node[state,accepting] (p4) [right=of p3] {$[k]$};

	\path[->] (q0) edge node [above] {$a\slash\lambda$} (q1)
			edge [loop above] node {$\sigma\slash\sigma$} ()
			edge node [sloped,above] {$\sigma\slash\tau$} (p1)
		(q1) edge node [above] {$\sigma\slash\lambda$} (q2)
			edge node [sloped,above] {$\sigma\slash\sigma \mid_{\sigma\neq a}$} (p1)
			edge node [sloped,above] {$\lambda\slash\sigma \mid_{\sigma\neq a}$}
				node [sloped,below] {$\sigma\slash\tau  \mid_{\tau\neq a}$} (p2)
		(q2) edge node [above] {$\sigma\slash\lambda$} (q3)
			edge node [sloped,above] {$\sigma\slash\sigma \mid_{\sigma\neq a}$} (p2)
			edge node [sloped,above] {$\lambda\slash\sigma \mid_{\sigma\neq a}$}
				node [sloped,below] {$\sigma\slash\tau \mid_{\tau\neq a}$} (p3)
		(q3) edge node [above] {$\sigma\slash\lambda$} (q4)
			edge node [sloped,above] {$\lambda\slash\sigma \mid_{\sigma\neq a}$}
				node [sloped,below] {$\sigma\slash\tau \mid_{\tau\neq a}$} (p4)
		(q4) edge node [sloped,above] {$\sigma\slash\sigma \mid_{\sigma\neq a}$} (p4)
		(p1) edge node [below] {$\sigma\slash\lambda\quad\lambda\slash\sigma$}
				node [above] {$\sigma\slash\tau$} (p2)
			edge [loop below] node {$\sigma\slash\sigma$} ()
		(p2) edge node [below] {$\sigma\slash\lambda\quad\lambda\slash\sigma$}
				node [above] {$\sigma\slash\tau$} (p3)
			edge [loop below] node {$\sigma\slash\sigma$} ()
		(p3) edge node [below] {$\sigma\slash\lambda\quad\lambda\slash\sigma$}
				node [above] {$\sigma\slash\tau$} (p4)
		(p4) edge [loop below] node {$\sigma\slash\sigma$} ();
\end{tikzpicture}
\parbox{4.3in}{\caption{A segment of the
input-altering transducer $\tsid$: for each $a\in \Sigma$ the complete transducer has $k$ states of the form $[i,a]$.
The labels $\sigma$ and $\tau$ on an edge mean: one edge for each $\sigma,\tau\in\Sigma$ with $\sigma\neq \tau$; for some edge sets additional restrictions apply denoted, for example, by $\mid_{\sigma\neq a}$.}\label{fig:inaltertrans}}
\end{figure}

\begin{definition}\label{def:tr}
The transducer
$$\tsid = (Q,\al,\al,E,[0],F)$$
is defined as follows. The set of states is
\[
	Q = \sett{[i]}{0\le i\le k}\cup \sett{[i,a]}{1\le i\le k, a\in \Sigma}
\]
with all but the initial state $[0]$ being final states: $$F = Q\sm\set{[0]}.$$
The edges in $\tsid$ can be divided into the four sets of edges $E = E_0 \cup E_s \cup E_i \cup E_d$.
The edges from $E_0$ do not introduce any error, edges from the other sets model one substitution ($E_s$), insertion ($E_i$), or deletion ($E_d$):
\begin{align}
	E_0 = &\bsett{[i] \xra{\sigma\slash \sigma}[i]}{\sigma\in\Sigma, 0\le i\le k} \cup {} \label{eq:nerr1} \\
		&\bsett{[i,a] \xra{\sigma\slash\sigma}[i]}{a,\sigma\in\Sigma,a\neq \sigma,1\le i\le k} \label{eq:nerr2} \\
	E_s = &\bsett{[i] \xra{\sigma\slash\tau}[i+1]}{\sigma,\tau\in\Sigma,\sigma\neq \tau,0\le i < k} \cup {} \label{eq:sub1} \\
		&\bsett{[i,a] \xra{\sigma\slash \tau}[i+1]}{a,\sigma,\tau\in\Sigma,\sigma\neq \tau,a\neq \tau,1\le i < k} \label{eq:sub2} \\
	E_i = &\bsett{[i] \xra{\e\slash \sigma}[i+1]}{\sigma\in\Sigma,1\le i< k} \cup {} \label{eq:ins1} \\
		&\bsett{[i,a] \xra{\e\slash \sigma}[i+1]}{a,\sigma\in\Sigma,a\neq \sigma,1\le i < k} \label{eq:ins2} \\
	E_d = &\bsett{[0] \xra{a\slash \e}[1,a]}{a\in\Sigma} \cup {} \label{eq:del0} \\
		&\bsett{[i] \xra{\sigma\slash \e}[i+1]}{\sigma\in\Sigma,1\le i< k} \cup {} \label{eq:del1} \\
		&\bsett{[i,a] \xra{\sigma\slash\e}[i+1,a]}{a,\sigma\in\Sigma,1\le i < k} \label{eq:del2}
\end{align}
\end{definition}
\pmsn \textbf{Terminology}. If $\tr=(Q,\al,\al,E,q_0,F)$ is a transducer in standard form, then we write $\etr$ for the NFA $$\etr=(Q,\ealph,E,q_0,F)$$
over the edit alphabet $\ealph$, where the labels of
the transitions in $\tr$ are viewed as elements of $\ealph$. Note that, the label of a path $P$ in $\tr$ is a pair of words $(u/v)$, whereas the label of the
corresponding path in $\etr$, which we denote as $\epath$, is an edit string
$h$ such that $\inp(h)=u$ and $\out(h)=v$.
This type
of NFA is called an eNFA in \cite{KaKo:2004}.

\begin{lemma}\label{lem:Adi:props}
Let $k\in\N$ and let $u,v$ be words. The following statements hold true with respect to the transducer $\tsid$.
\begin{enumerate}[i.)]
\item
In $\etsid$, every path from the start state $[0]$ to any state $[i]$ or $[i,a]$ has as label a reduced edit string
whose weight is equal to $i$.
\item
If $1\le\dsid(u,v)\le k$ and $h$ is a reduced edit string realizing $\dsid(u,v)$, then
$h$ is accepted by $\etsid$.
\item
If $v\in\tsid(u)$, then $1\le\dsid(u,v)\le k$.
\item
If $\dsid(u,v)\le k$ and $va \le_p u$, for some symbol $a$, then   $[0]\xras{u\slash v} [\dsid(u,v),a]$.
\item
If $i\in\N$ and $i+\dsid(u,v)\le k$, then $[i]\xras{u\slash v} [i+\dsid(u,v)]$.
\end{enumerate}
\end{lemma}
\begin{proof}
The \underline{first} statement follows when we note that the  definition of $\tsid$ and $\etsid$ implies the
following facts: (a) An edge exists between a state
with error counter $i$ to one with error counter  $i+1$,
if and only if the label of that edge is an error;
thus, in any path from $[0]$ to $[i]$ or $[i,a]$, the label of that path consists of exactly $i$ errors.
(b) Any edit string accepted by $\etsid$ is indeed reduced.
\par
For the \underline{second} statement, consider any reduced edit string $h$ realizing $\dsid(u,v)$.
If the first error in $h$ is a deletion, then $h$ is of
the form
      \[
      h=(e_1\cdots e_r)(a/\e)(b_1/\e)\cdots(b_d/\e)h',
      \]
where each $e_i$ is a non-error edit operation
of the form $(\sigma_i/\sigma_i)$, $(a/\e)$
is a deletion error, $d\in\N_0$ and each $(b_j/\e)$ is
a deletion error, and $h'$ is an edit string that is
either empty
or starts with a non-deletion edit operation $(x/y)$
such that $y\not=a$. If $h'$ is nonempty, then by definition of $\etsid$ the
following is a path
\[
[0]\xras{(e_1\cdots e_r)}[0]\xras{(a/\e)(b_1/\e)\cdots(b_d/\e)}[1+d,a]
\xras{h'}[1+d+\weight{h'}]
\]
accepting $h$. Similarly, a path accepting $h$ exists in
$\etsid$ , if $h'$ is empty.
\par
Finally, one verifies that if the first error in $h$
is a substitution, then again $h$ is accepted by $\etsid$.
\par
For the \underline{third} statement, if $v\in\tsid(u)$, then $(u/v)$ is the label of a path $P$ from $[0]$
to a final state $[i]$ or $[i,a]$, with $0<i\le k$. As the label of the path $\epath$ has exactly $i$ errors,
it follows that $\dsid(u,v)\le i\le k$.
\par
 We also need to show that $\dsid(u,v)\ge1$, that is, $u\not=v$.
First consider the case where the path $P$ ends at $[i,a]$, with $1\le i\le k$. Then,
the label of $\epath$ is an edit string of the form
$$h=(\sigma_1/\sigma_1)\cdots(\sigma_r/\sigma_r)(a/\e)
(b_1/\e)\cdots(b_d/\e)$$
and $u=\inp(h)= \sigma_1\cdots \sigma_r a b_1\cdots b_d$ and $v=\out(h)=\sigma_1\cdots\sigma_r$.
Hence, $u\not=v$. Now consider the case where the path $P$ ends at state $[i]$.
There are three cases. (a) The states used in the path are $[0], [1],\ldots,[i]$.
(b) The states used in $P$ are $[0],[1,a],\ldots,[r,a],[r],\ldots,[i]$, for some appropriate $[r]$.
(c)  The states used in $P$ are $[0],[1,a],\ldots,[r,a],[r+1],\ldots,[i]$, for some appropriate $[r]$.
In all three cases, one verifies that $u\not=v$. For example, in case (b), $u$ must be of the form
$xa\sigma_1\cdots\sigma_{r-1}\sigma y$ and $v$ of the form $x\sigma z$, where the
$\sigma_j$'s are symbols, $x,y,z$ are words, and $\sigma$ is a symbol other than $a$; hence,
$u\not=v$.
\par
For the \underline{fourth} statement, let $u=a_1\cdots a_r ab_1\cdots b_t$, with each $a_i$ and $b_j$ being a symbol,
and $v=a_1\cdots a_r$.
We use lemma~\ref{lem:didist}.
The edit string $$h=(a_1/a_1)\cdots(a_r/a_r)(a/\e)(b_1/\e)\cdots(b_t/\e)$$
realizes $\dsid(u,v)$. Moreover, this edit string is the label of a path in $\etsid$ from
$[0]$ to $[\dsid(u,v),a]$. Hence, there is a path
in $\tsid$ from $[0]$ to $[\dsid(u,v),a]$ with
label $(\inp(h)/\out(h))=(u/v)$, as required.
\par
For the \underline{fifth} statement, let $h$ be an edit string realizing $\dsid(u,v)$. By definition
of $\etsid$, at each state of the form $[j]$ with $0<j<k$, one can follow an edge whose
label $e$ can be of any of the
four types of edit operations, and moreover, if $e$ is an error, then the edge goes into $[j+1]$,
that is, $[j]\xra{e}[j+1]$.
As $h$ contains exactly $\dsid(u,v)$ errors, there is a path from $[i]$ to $[i+\dsid(u,v)]$ whose label is made of the edit operations in $h$. Hence, there is a path in $\tsid$ from $[i]$ to $[i+\dsid(u,v)]$ whose label is $(u/v)$, as required.
\end{proof}

\begin{theorem}\label{th:iat:ed}
For each $k\in\N$, the transducer $\tsid$ is input-altering and of size $O(k)$, and
describes the property of error-detection for the channel $\chsid(k)$.
\end{theorem}
\begin{proof}
By construction, it follows that $\tr_k$ is trim and has a number of states and transitions that is linear with respect to $k$. Hence, it is indeed of size $O(k)$.
The third statement of lemma~\ref{lem:Adi:props} implies that the transducer is input-altering.
For the error-detection part, using the first statement of remark~\ref{rem:ed}, it is sufficient  to show that, for every language $L$,
\[
\tsid(L)\cap L=\emptyset \>\hbox{ if and only if }\> \dsid(L)>k.
\]

First, for the `if' part, assume $\dsid(L)>k$ and consider any words $u,v\in L$.  We need to
prove $v\notin\tsid(u)$.
If $u=v$ then this holds as $\tsid$ is input-altering. Else, it follows from the third statement of
lemma~\ref{lem:Adi:props}.  Now for the `only if' part, assume
\begin{equation}\label{eq:iatp}
\tsid(L)\cap L=\emptyset,
\end{equation}
but, for the sake of contradiction,  suppose there are different words $u,v\in L$ such that $1\le\dsid(u,v)\le k$. 
If $v$ is a prefix of $u$, then
$va\le_p u$, for some $a\in\al$, and
the fourth statement of the above lemma implies
$[0]\xras{u\slash v} [\dsid(u,v),a]$ and, therefore, $v\in\tsid(u)$, which contradicts~(\ref{eq:iatp}).
By symmetry, a contradiction arises if $u$ is a prefix of $v$.

Now consider the case where $v$ is not a prefix of  $u$, and $u$ is not a prefix of $v$. Then, $u=xau'$ and $v=xbv'$ for some words $x,u',v'$ and   symbols $a,b\in\al$ with $$a\not=b.$$
\emph{We shall obtain a contradiction
to (\ref{eq:iatp}) by showing the existence of a path $[0]\xras{u\slash v}\psi$, or $[0]\xras{v\slash u}\psi$, where $\psi$ is a final state of $\tsid$}.
Let $h$ be an edit string realizing
$\dsid(au',bv')$. Recall $\dsid(u,v)=\dsid(au',bv')$. As $a\not=b$, the first edit operation, say  $e$, of $h$ must
be an error, that is, not of the form $\sigma/\sigma$. Let $h=eh'$. We consider three cases for $e$.
First, if $e$ is a substitution, then $e=(a/b)$  and $\dsid(u,v)=1+\dsid(u',v')$. By the fifth statement of the above lemma, $[1]\xras{u'\slash v'} [1+\dsid(u',v')]$. Then, the required path is
\[
[0]\xras{x/x}[0]\xras{a/b}[1]\xras{u'\slash v'} [1+\dsid(u',v')].
\]

Now consider the case where $e=(a/\e)$. Then, $\dsid(u,v)=1+\dsid(u',bv')$ and $h'$ realizes $\dsid(u',bv')$. Let $d$ be the number of deletions (if any) at the beginning of $h'$ so that any edit operation
following these deletions is not a deletion. Thus, $h'$ is of the form $h_1h_2$ with $\inp(h_1)=u_1$
and $\out(h_1)=\e$, where $u_1$ is a word of length $d$, and
$\inp(h_2)=u_2$ and $\out(h_2)=bv'$, for some word $u_2$, and
$u'=u_1u_2$,
and $\dsid(u',bv')=d+\dsid(u_2,bv')$.
As $bv'$ is nonempty, also $h_2$ is nonempty, so let $e'$ be the first edit operation of $h_2$, which
cannot be a deletion.
If $u_2=\e$, then $e'=(\e/b)$ and $h_2$ consists of insertions, and the required path is
\[
[0]\xras{x/x}[0]\xras{a/\e}[1,a]\xras{u_1\slash \e}[1+d,a]\xras{e'} [1+d+1]\xras{\e\slash v'} [1+d+1+\dsid(\e,v')].
\]
If $u_2\not=\e$, then there is a symbol $c$ such that $u_2=cu_2'$.
If $c=b$, then  $e'$ cannot be a substitution, so it must be the non-error $(c/b)$ or the insertion
$(\e/b)$. Then, the required path is
\[
[0]\xras{x/x}[0]\xras{a/\e}[1,a]\xras{u_1\slash \e}[1+d,a]\xras{e'}[1+d+t]
\xras{z\slash v'} [1+d+t+\dsid(z,v')],
\]
where $z=u_2'$ and $t=0$ (case of $e'=(c/b)$), or $z=cu_2'$ and $t=1$ (case of $e'=(\e/b)$).
If $c\not=b$, then $e'$ must be the insertion $(\e/b)$ or the substitution $(c/b)$. Again, in either case,
a path as required exists.

Finally, the case of $e=(\e/b)$, is symmetric to the previous one by simply switching the
roles of $u$ and $v$.
\end{proof}

\subsection{The $O(n^2d)$ algorithm for edit distance}
Here we use the results of the previous subsection to arrive at an efficient algorithm for computing the desired
edit distance.
Remark~\ref{rem:iat}  and theorem~\ref{th:iat:ed}  imply that the \emph{intermediate} algorithm
\texttt{DistFirstInpAlter}
shown below correctly computes the desired edit distance. Moreover, by reasoning as in the proof of corollary~\ref{th:ed}, it follows that this algorithm executes in time $O(\sz{\aut}^2r^2\db_{\aut}\log\db_{\aut})$, where $r$ is the cardinality of the alphabet used in $\aut$.

\begin{figure}[ht]
\begin{tabbing}
PAR \= No \= No \= NN \= NN \=\kill
\> Algorithm \texttt{DistFirstInpAlter} \\
\> 0.\> Input: NFA $\aut$ \hspace{4mm} \\
\> 1.\> Let ${\db_\aut}$ be the bound in Lemma~\ref{lem:bound}\\
\> 2.\> Let $\min\leftarrow 1$ and $\max\leftarrow{\db_\aut}-1$ \\
\> 3.\> Perform binary search to find  the largest $k$ in
      $\{\min,\ldots,\max\}$ \\
\>  \> for which $L(\aut)$ is error-detecting for $\chsid(k)$ as follows: \\
\>  \>  \textbf{while} ($\min\le \max$)\\
\>  \>  a)\> Let $k\leftarrow\lfloor(\min+\max)/2\rfloor$\\
\>   \> b)\> Construct the transducer $\tsid$ \\
\>   \> c)\> Construct NFA  $\aut'$ accepting $\tsid(L(\aut))\cap L(\aut)$\\
\>   \> d)\> If ($\aut'$ accepts $\emptyset$)  let $\min\leftarrow k+1$\\
\>   \> \>  Else  let $\max\leftarrow k-1$\\
\> 4. \> \textbf{return} $\min$
\end{tabbing}
\end{figure}

We note again that, in the worst case, $\db_{\aut}$ is
of order $O(\sz{\aut})$ and, assuming a fixed alphabet, the above algorithm operates in time
$$O(\sz{\aut}^3\log\sz{\aut}),$$
which is asymptotically better than those of all other
known algorithms. However, we now discuss in detail the second major observation stated in the beginning of section~\ref{sec:iat}, which leads to the most efficient algorithm in theorem~\ref{th:final}. In particular, for the sake of clarity, we present that algorithm in two steps.
In the first place, we notice that the while loop
in \texttt{DistFirstInpAlter}
can be replaced with the construction
of the automaton $\tsidp{{\db_\aut}-1}(\aut)\cap\aut$ and a search
in that automaton for a path from the start state to
a final one in which the error counter value is minimal (this
value would be the required edit distance).
\begin{tabbing}
PAR \= No \= No \= NN \= NN \=\kill
\> Algorithm \texttt{DistNextInpAlter} \\
\> 0.\> Input: NFA $\aut$ \hspace{4mm} \\
\> 1.\> Let ${\db_\aut}$ be the bound in Lemma~\ref{lem:bound}\\
\> 2.\> Construct the transducer $\tsidp{{\db_\aut}-1}$ \\
\> 3.\> Construct NFA  $\aut'$ accepting
        $\tsidp{{\db_\aut}-1}(L(\aut))\cap L(\aut)$\\
\> 4.\> Starting at the start state of $\aut'$, use breadth first search (BFS) \\
\>   \> to visit all states. In doing so, keep track of the smallest\\
\>   \> error counter $\min$ in the visited final states of $\aut'$.\\
\> 5. \> \textbf{return} $\min$
\end{tabbing}
As usual in product constructions, the states of $\aut'$ are triples of the form
$(\phi,q,q')$, where $\phi$ is a state of $\tsidp{{\db_\aut}-1}$, and $q$, $q'$ are states of $\aut$. The start state of $\aut'$ is
$([0],q_0,q_0)$, where $q_0$ is the start state of $\aut$, and the final states of $\aut'$ are those triples consisting of final states in $\tsidp{{\db_\aut}-1}$ and $\aut$. A transition
\[
(\phi,q,q')\xra{y}(\psi,r,r')
\]
exists in $\aut'$ if and only if the following transitions
\[
\phi\xra{x\slash y}\psi,\quad
q\xra{x} r,\quad
q'\xra{y} r'
\]
exist in $\tsidp{{\db_\aut}-1}$, $\aut^\e$ and $\aut^\e$, respectively,
for some label $x$, where $\aut^\e$ results if we add to
$\aut$ empty  loop transitions $(q,\e,q)$ for all states
$q$ in $\aut$. The correctness of the above algorithm follows
from lemma~\ref{lem:Adi:props} and the definition of
$\aut'$. The breadth first search process requires time
linear with respect to the size of $\aut'$, which is
\[
O(\sz{\aut}^2\db_{\aut}),
\]
and this also is the time complexity of the above
algorithm (when the alphabet is fixed).
\par
The final improved algorithm results if we notice that
the desired edit distance can be much smaller
than $\db_{\aut}$ and that it can be computed
using only an `initial' part of $\tsidp{{\db_\aut}-1}$.
In other words, one can first build $\tsidp{1}$ and
$\aut'_1$ accepting $\tsidp{1}(L(\aut))\cap L(\aut)$, and test whether $\aut'_1$ has any accepting path. If not, this
process is repeated by extending $\aut'_k$ to $\aut'_{k+1}$
until some extended automaton, say $\aut'_d$, has an
accepting path, in which case the desired distance is
equal to $d$.
\begin{tabbing}
PAR \= No \= No \= NN \= NN \=\kill
\> Algorithm \texttt{DistBestInpAlter} \\
\> 0.\> Input: NFA $\aut$ \hspace{4mm} \\
\> 1.\> Construct the transducer $\tsidp{1}$ \\
\> 2.\> Construct NFA  $\aut'$ accepting $\tsidp{1}(L(\aut))\cap L(\aut)$ \\
\> 3.\> $k\leftarrow1$ \\
\> 4. \>  \textbf{while} ($\aut'$ has no accepting path)\\
\>  \>  a)\> $\aut'\leftarrow \texttt{Extend}(\aut',k)$\\
\>   \> b)\> $k\leftarrow k+1$  \\
\> 5. \> \textbf{return} $k$
\end{tabbing}
The function \texttt{Extend} in the above algorithm works
based on the structure of $\tsid$ in Fig.~\ref{fig:inaltertrans} and is \emph{partially} shown below. For clarity,
we emphasize the fact that, in each step $k$ of this algorithm, the final states of $\aut'$ are only
triples of the form $([k,a],f,f')$ or $([k],f,f')$, that is,
when $i<k$, no triples of the form $([i,a],f,f')$ or $([i],f,f')$ are final states in $\aut'$.
\begin{tabbing}
PAR \= No \= No \= NN \= NN \=NN \=\kill
\> Function \texttt{Extend}$(\aut',k)$ (partial view)\\
\> let $\autb $ be a copy of $\aut'$ \\
\> for each state of the form $([k,a],q,q')$ in $\aut'$ \\
\> \> for each transitions $q\xra{\sigma}r$ and $q'\xra{\sigma'}r'$ in $\aut$ \\
\>\>\> if ($a\not=\sigma'$ and $\sigma\not=\sigma'$)\\
\>\>\>\> add to $\autb $ the transition $([k,a],q,q')\xra{\sigma/\sigma'}([k+1],r,r')$\\
\>\>\>\> if $r$ and $r'$ are final in $\aut$ then
$([k+1],r,r')$ is final in $\autb $\\
\>\>\> if ($a\not=\sigma'$ and $\sigma=\sigma'$)\\
\>\>\>\> add to $\autb $ the transition $([k,a],q,q')\xra{\sigma/\sigma}([k],r,r')$\\
\>\>\>\> if $r$ and $r'$ are final in $\aut$ then
$([k+1],r,r')$ is final in $\autb $\\
\>\>\> $\cdots\cdots\cdots$\\
\> \textbf{return} the NFA $\autb $
\end{tabbing}
Based on the above discussion, the correctness of the following theorem has been established.
\begin{theorem}\label{th:final}
Algorithm \texttt{DistBestInpAlter} computes the
edit distance of the language given via an NFA $\aut$ in
time $O(\sz{\aut}^2r^2d)$, where $r$ is the cardinality
of the alphabet used in $\aut$ and $d$ is the
computed edit distance.
\end{theorem}
%

\section{Implementation and testing}\label{sec:implem}
We have implemented the main algorithm \texttt{DistBestInpAlter} of theorem~\ref{th:final}, as well as the intermediate versions
\pssi
\texttt{DistErrDetect, DistErrCorrect, DistFirstInpAlter,}
\pssn
using the FAdo library for automata
\cite{Fado}, which is well maintained and provides
several useful tools for manipulating automata.
Moreover, we have used some of the implementations
of I-LaSer \cite{ilaser} involving product constructions
between transducer and automaton objects of the FAdo
library. We note that an implementation in C++ of the
algorithm in \cite{Kon:2007} is discussed in
\cite{Daka:2011}, but the execution time is too slow
to be used for any meaningful comparisons with the algorithms
presented here. Our best algorithm can be executed online
at \cite{olaser}. The user can enter as input
an NFA in Grail or FAdo format, select the algorithm to execute,
and press the Submit button.

We have performed several
tests\footnote{All tests were performed on a machine with the
following specification.
  Make: Acer,
  CPU: AMD Athlon(tm) II X2 215,
  Clock speed: 2.70 GHz,
  Memory (RAM): 4.00 GB,
  Operating System: Windows 7 64-bit.}
for the correctness
of these algorithms, as well as two sets of tests for the time complexity, which confirm  the theoretical
result that \texttt{DistBestInpAlter}
is indeed the fastest algorithm. We note that, for the
other three algorithms, we have skipped the step of
computing the upper bound $\db_{\aut}$ on the edit distance,
as this step is the same for all these algorithms, thus
resulting in faster execution without affecting in any
essential way the performance comparisons.

The two sets of tests correspond to two sequences of
automata $(\aut_n)$ and $(\autb_n)$, shown in the next two figures, for which we used $n$ as the value of $\db_{\aut}$.
The first test set is such that the desired distance is equal to $n$, for each NFA $\aut_n$, that is, the distance grows
with $n$ and, in fact, it is a worst-case scenario where the
distance is equal to the number of states of the NFA.
The second test set is such that the desired distance is fixed, equal to 2,
for all $n$.
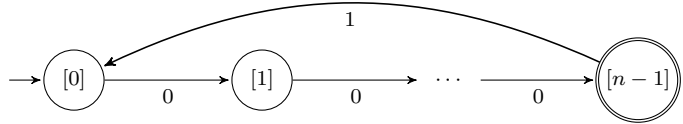
\begin{figure}[ht]
\centering
\begin{tikzpicture}[shorten >=1pt,node distance=2.5cm,
		on grid,>=stealth',initial text=,
		every state/.style={inner sep=2pt, minimum size=.8cm},
		 phantom/.style={draw=white},
to/.style={->,>=stealth',shorten >=1pt,semithick,font=\sffamily\footnotesize},
        font=\footnotesize]
	\node[state,initial] (q0) {$[0]$};
	\node[state] (p1) [right=of q0] {$[1]$};
	\node[state,phantom] (p2) [right=of p1] {$\cdots$};
	\node[state,accepting] (p3) [right=of p2] {$[n-1]$};

	\path[->] (q0) edge node [below] {$0$} (p1)
		 (p1) edge node [below] {$0$} (p2)
		 (p2) edge node [below] {$0$} (p3);
    \draw[to]
		(p3) to[bend right=25] node[midway,below] {$1$}
        (q0);
\end{tikzpicture}
\parbox{4.3in}{\caption{The automaton $\aut_n$ accepting the language $0^{n-1}(10^{n-1})^*$.}\label{fig:testA}}
\end{figure}

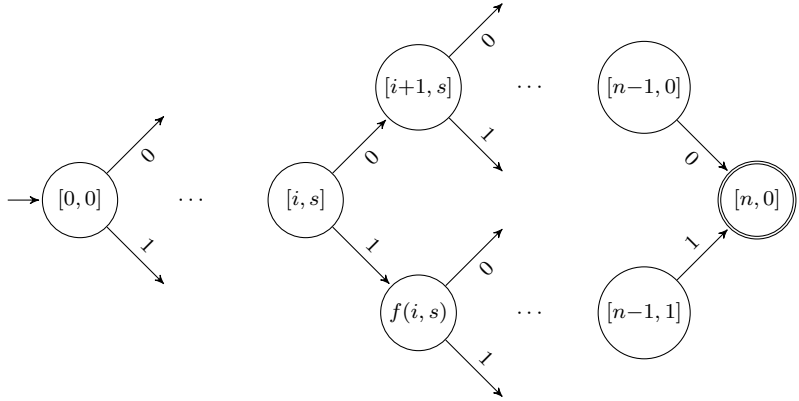
\begin{figure}[ht]
\centering
\begin{tikzpicture}[shorten >=1pt,node distance=1.5cm and 1.5cm,
		on grid,>=stealth',initial text=,
		every state/.style={inner sep=2pt, minimum size=1cm},
		phantom/.style={draw=white},
		to/.style={->,>=stealth',shorten >=1pt,semithick,font=\sffamily\footnotesize},
        font=\footnotesize]
	\node[state,initial] (q0) {$[0,0]$};
	\node[state,phantom] (q1) [above right=of q0] { };
	\node[state,phantom] (q2) [below right=of q0] { };
	\node[state,phantom] (q3) [right=of q0] {$\cdots$};

	\node[state] (p0) [right=of q3] {$[i,s]$};
	\node[state] (p1) [above right=of p0] {$[i{+}1,s]$};
	\node[state,phantom] (p11) [above right=of p1] { };
	\node[state,phantom] (p12) [below right=of p1] { };
	\node[state] (p2) [below right=of p0] {$f(i,s)$};
	\node[state,phantom] (p21) [above right=of p2] { };
	\node[state,phantom] (p22) [below right=of p2] { };
	\node[state,phantom] (p3) [right=of p1] {$\cdots$};
	\node[state,phantom] (p4) [right=of p2] {$\cdots$};

	\node[state] (r1) [right=of p3] {$[n{-}1,0]$};
	\node[state] (r2) [right=of p4] {$[n{-}1,1]$};
	\node[state,accepting] (r3) [below right=of r1] {$[n,0]$};

	\path[->]
		(q0) edge node [below,sloped] {$0$} (q1)
			 edge node [above,sloped] {$1$} (q2)
		(p0) edge node [below,sloped] {$0$} (p1)
			 edge node [above,sloped] {$1$} (p2)
		(p1) edge node [below,sloped] {$0$} (p11)
			 edge node [above,sloped] {$1$} (p12)
		(p2) edge node [below,sloped] {$0$} (p21)
			 edge node [above,sloped] {$1$} (p22)
		(r1) edge node [below,sloped] {$0$} (r3)
		(r2) edge node [above,sloped] {$1$} (r3);

\end{tikzpicture}
\parbox{4.3in}{\caption{The automaton $\autb_n$ with $n^2+n+1$ states, where
$f(i,s)=[i+1,(s+i+1)\%(n+1)]$. The states are $[n,0]$
and $[i,s]$, with $0\le i\le n-1$ and $0\le s\le n$.
This automaton accepts the Levenshtein code consisting of all binary words $b_1\cdots b_n$ of length $n$ such that $(\sum_{i=1}^ni\cdot b_i)\%(n+1)=0$, where `\%' is the integer division remainder operation. This code has edit distance equal to 2. On the other hand, its distance for insertion/deletion errors only is 3, so it is error-correcting for the 1-insertion/deletion per word channel.}\label{fig:testB}}
\end{figure}
The next table shows the actual running times of the
four algorithms on the NFAs
$\aut_4,\ldots,\aut_8,\aut_{13}, \aut_{21}, \aut_{31}$.
The number in parentheses next to each $\aut_i$ indicates the number of states in $\aut_i$.

\pmsi
\begin{tabular}{|c|c|c|c|c|}\hline
NFA & ErrDetection & ErrCorrection  &  FirstInpAlter
& BestInpAlter\\  \hline
$\aut_5\>(5)$ &  3.94s & 0.35s & 0.08s & 0.008s  \\  \hline
$\aut_6\>(6)$ & 19.20s & 0.48s & 0.11s & 0.010s  \\  \hline
$\aut_7\>(7)$ & 107.35s & 2.54s & 0.18s & 0.013s  \\  \hline
$\aut_8\>(8)$ & 442.01s & 4.03s & 0.33s & 0.016s  \\  \hline
$\aut_{13}\>(13)$ & $> 5 \hbox{ hours}$ & 144.75s & 1.31s & 0.020s  \\  \hline
$\aut_{21}\>(21)$ & $> 5 \hbox{ hours}$ & 12475.27s & 10.21s & 0.029s  \\  \hline
$\aut_{31}\>(31)$ & $> 5 \hbox{ hours}$ & $> 5 \hbox{ hours}$ & 46.28s & 0.109s  \\  \hline
\end{tabular}
\pbsn
The next table shows the actual running times of the
four algorithms on the NFAs $\autb_3,\ldots,\autb_8$.
The number in parentheses next to each $\autb_i$ indicates the number of states in $\autb_i$.
\pmsi
\begin{tabular}{|c|c|c|c|c|}\hline
NFA & ErrDetection & ErrCorrection  &  FirstInpAlter
& BestInpAlter\\  \hline
$\autb_3\>(13)$ & 0.889s & 0.164s & 0.098s & 0.027s  \\  \hline
$\autb_4\>(21)$ & 212.32s & 7.06s & 0.655s & 0.039s  \\  \hline
$\autb_5\>(31)$ & $>5 \hbox{ hours}$ & 72.25s & 4.63s & 0.097s  \\  \hline
$\autb_6\>(43)$ & $>5 \hbox{ hours}$ & 3806.74s & 40.79s & 0.234s  \\  \hline
$\autb_7\>(57)$ & $>5 \hbox{ hours}$ & $>5 \hbox{ hours}$ & 375.17s & 0.735s  \\  \hline
$\autb_8\>(73)$ & $>5 \hbox{ hours}$ & $>5 \hbox{ hours}$ & 2070.21s & 1.919s  \\  \hline
\end{tabular}
\pbsn
Let $d$ be the computed edit distance in the above test
sets.
The best algorithm has about the same performance in both test cases, even though $d$ is a parameter in its time complexity $O(n^2 d)$.
A possible explanation is that the NFA $\autb_i$ has more edges than those of $\aut_j$, when both $\autb_i$ and $\aut_j$ have the same number of states.
A reason why the intermediate algorithms perform a lot better on  $\autb_i$ than on $\aut_j$ with the same number of states is that the value of the edit distance upper bound
in $\autb_i$ is smaller than that in $\aut_j$.
\pbsn
A further improvement to our algorithms, to be implemented, is that one can
remove from Fig.~\ref{fig:inaltertrans} all the diagonal transitions
from a state $[i,a]$ to a state $[i+1]$. This is because,
for any edit string of the form
\[
h=e_1\cdots e_r\>(a/\e)(a_1/\e)\cdots(a_d/\e)(\sigma/\tau)\>h_1
\]
accepted by $\etr_k$, where $\tau\notin\{a,\sigma\}$ and
the $e_j$'s are non-errors, the
automaton $\etr_k$ also accepts
\[
g=e_1\cdots e_r\>(a/\tau)(a_1/\e)\cdots(a_d/\e)(\sigma/\e)\>g_1
\]
such that $\inp(g)=\inp(h)$, $\out(g)=\out(h)$, and
$\weight{g}=\weight{h}$. Moreover $\etr_k$ accepts
$g$ using none of the diagonal transitions that are
to be removed as specified above.
A similar observation holds if we replace in $h$ the edit
operation $(\sigma/\tau)$ shown in $h$ with $(\e/\sigma)$,
where $\sigma\not=a$.
Of course this improvement does not affect in any significant
way the performance comparisons among the four algorithms.

\section{Conclusion}\label{sec:last}
This paper represents a significant improvement in  the
time complexity of computing the edit distance of a given regular language.
As discussed in \cite{Kon:2007}, this problem is related to
the inherent capability of a language to detect substitution,
insertion, and deletion errors. The method based on using the input-altering transducer $\tsid$
seems to adapt to other types of errors as well.
For example, one can construct a similar input-altering transducer for insertion/deletion only errors. It seems
promising to investigate the problem when the errors
have different costs (in the current setting, the cost of
each error is 1).

The present contribution stemmed from the
question of whether the error-detection property for
the channel  $\chsid(k)$ can be described by an
input-altering transducer. The more general question
of whether an input-preserving transducer property of
interest can be described by an input-altering transducer is important to investigate, as this would lead to more
efficient algorithms for deciding the property
satisfaction problem (whether, given a regular language
and a transducer property, the language satisfies the property).


\bibliographystyle{abbrv}
\bibliography{refs}

\end{document}